\documentclass[a4paper]{article}
\usepackage[applemac]{inputenc}
\usepackage{amsmath,latexsym, eucal}
\usepackage{amssymb,amsthm}
\usepackage{pgf}
\usepackage{tikz}
\usetikzlibrary{arrows,automata}
\usepackage{graphicx}
\usepackage{authblk}
\DeclareGraphicsExtensions{.eps,.pdf,.png,.jpg}
\newtheorem{definition}{Definition}
\newtheorem{thm}{Theorem}

\begin{document}

\title{Personal data disclosure and data breaches: the customer’s viewpoint}
\author[1]{Giuseppe D’Acquisto}
\author[2]{Maurizio Naldi}
\author[3]{Giuseppe F. Italiano}
\affil[1]{Garante per la Protezione dei Dati Personali, Piazza di Monte Citorio n. 121, Rome, Italy}
\affil[2]{Universit\`{a} di Roma Tor Vergata, Dipartimento di Informatica Sistemi Produzione, Via del Politecnico 1, 00133 Roma, Italy}
\affil[3]{Universit\`{a} di Roma Tor Vergata, Dipartimento di Informatica Sistemi Produzione, Via del Politecnico 1, 00133 Roma, Italy}
\date{\today}
\maketitle
\tableofcontents

\newpage
\section{Introduction}
\label{Intro}
In most cases, a customer can receive a service over the Internet without disclosing its personal data, aside from its IP address. However, some services may need personal information or may be accomplished faster and more easily if the customer is willing to release some personal information. In other cases the service may be enhanced if more is known about the customer. For example an online store may propose additional items for purchase if the customer's purchasing history is known, or if the customer's tastes are known otherwise, and a cloud computing provider may sell additional CPU time or storage space if the consumption behaviour of the customer is known. The effect of service enhancements as incentives for the customer to loose its privacy requirements has been known since long as a marketing tool \cite{Milne}. In all those cases the customer is induced into divulging some personal information, so that the customer and the service provider share such information. However, the release of personal information carries along negative consequences as well when it leaks outside the customer-service provider circle. Such unintentional release of secure information to an untrusted environment is called \textit{data breach}. When a data breach occurs, that information falls in the hands of a third party, which may accomplish an identity theft or exploit those data for malicious uses (e.g., fraudulent unemployment claims,  fraudulent tax returns, fraudulent loans, home equity fraud, and payment card fraud, as briefly reported in \cite{Romanosky}). Such leakages may turn into an economical loss both for the service provider that had to protect the customer's data and for the customer as well \cite{Ponemon}. A service provider has then to perform a trade-off between its investments in security and the losses it may incur because of data breaches, and may be led to invest in privacy-enhancing technologies \cite{LE2010}. On the other hand, the customer has to balance the benefits obtained by releasing its personal information against the potential losses associated to a data breach. In the following the term customer is not necessarily referred to an individual, but includes companies that act as customers, for which the amount of money at risk through data breaches is quite larger than for individuals. 

Although some effort has been spent to determine optimal investment policies for a company investing in security on its information system \cite{Cavusoglu,Gordon}, to the best of our knowledge no aids have been proposed to help the customer's decision, when the customer heavily relies on the security of its service provider's information system. However, it has been argued that the economic analysis of the trade-offs that a company faces when tracking its customers is gaining relevance, since some start-ups have started offering bargains in return for users’ data \cite{Krishna}. Such offers strengthen the awareness that customers' personal information have a value, and spur the customer themselves to carefully assess if the release of those information is worthwhile. 

Our aim is to explore the pros and cons the customer gets when releasing personal data and facing the danger of identity theft. In this paper we introduce a model to describe the advantages and disadvantages associated with the release of personal data, and formulate the customer's problem of choosing the right level of personal information release as a trade-off problem, with the aim of maximizing the customer's surplus. Rather than considering just the service provider as the party in charge of protecting the data, our models consider the wider viewpoint that data breaches may occur on the service provider's side as well as on the customer's side. We prove that the solution to the trade-off problem exists and is unique, and show that solution in a reference scenario. The results represent a first step to model, in a parsimonious way, the complex interaction between service providers and customer as to demand, security, and privacy, from an economical perspective. It appears that, among the parameters outside the control of the customer, the customer's decision is most sensitive to the price imposed for the service and to the impact of the profiling capability of the service provider on the loss incurred after data breaches. Finally, we consider the limit case of a perfectly secure service provider, so that data breaches can occur just on the customer's side. For that case we provide closed form expressions both for the optimal level of exposure for the customer, and for the sensitivity with respect to all the parameters involved in the trade-off problem (embodied by the elasticity or quasi-elasticity functions). We believe that the results may provide theoretical grounds to introduce a regulatory approach to the issue, to achieve a regulated balance of interests between the contrasting aims of preserving the customer's privacy and spurring the service provider's business.

\section{Service demand curve, data disclosure and the effects of data breaches}
\label{Demand}
When releasing personal data in exchange for enhanced services, the customer has to perform a trade-off between what it releases, the benefits it gets, and the risk related to divulging that information. The main variables involved in that trade-off are: the quantity of services that are sold; their unit price; the personal information that are released as part of the service sale; and the potential loss deriving from that information disclosure. In order to derive some guidelines for the customer releasing that information, we need to define the relationship between these quantities. In this section, we provide analytical models that link all the four quantities mentioned above. 

We recall that the service provider sells services at a unit price $p$; the customer buys a quantity $q$ of such services, represented, e.g., by minutes of phone traffic, bytes of data traffic volume, digital units, CPU time, bytes of storage capacity. The relationship between $p$ and $q$ is the \textit{demand curve} \cite{Mankiw}. For sake of simplicity, we assume here that in our case the relationship is linear. When no personal information is disclosed, those quantities are then related by the expression
\begin{equation}
\label{Dem}
\frac{q}{q^{*}}+\frac{p}{p^{*}}=1,
\end{equation} 
where $q^{*}$ is the maximum quantity of service that the service provider can provide, and $p^{*}$ is the maximum unit price that the customer can sustain (a.k.a. its \textit{willingness-to-pay}). When the service is free ($p=0$), the customer asks for the maximum quantity that the service provider can supply ($q=q^{*}$). When the price is larger than the willingness-to-pay ($p\ge p^{*}$), the customer will not buy the service, and the quantity of service sold will be $q=0$. In the following we treat both the quantity of service $q$ and the unit price $p$ as continuous variables (though their variation is actually discrete), since we assume that their granularity is extremely small with respect to the values at hand.

However, if the customer is willing to release some personal information, the service provider eases the provision of services, e.g., by providing personalized services or automatic login. In fact, the more the service provider knows about the customer, the better it (or any other third party associated to the service provider and sharing that information) can shape and direct its offer to achieve a sale. The release of personal data can therefore help reduce the product/service search costs for both parties: the time employed by customers when looking for that product/service, and the effort spent by sellers trying to reach out to their customers. Varian has shown that customers rationally want some of their personal information to be available to sellers \cite{Varian-09}. Hence, the customer is incentivized to supply its personal data and increase its consumption. We assume that the overall information that the customer can release is collected into a number $N$ of sets $\{I_{1}, I_{2}, \ldots, I_{N}\}$, so that each set includes the information belonging to the previous set, i.e., $I_{i}\subset I_{i+1}$, with $i=1,2,\ldots , N-1$.

Each release of information by the customer is rewarded by a new offer by the service provider, which at the same time incentivizes the consumption. The demand curve correspondingly changes, e.g., as illustrated in Figure~\ref{fig:Domanda}, where we can visualize the shift in the customer's consumption by observing how the working point moves onto the new demand curve. For example, we can consider in Figure~\ref{fig:Domanda} the point $(q_{1},p_{1})$ on the pre-release demand curve represented by Eq. (\ref{Dem}). After the release of personal data the customer may move to the working point $(q_{2},p_{2})$ on the after-release demand curve. 
 
\begin{figure}[tp]
\begin{center}	
  \includegraphics[width=.8\columnwidth]{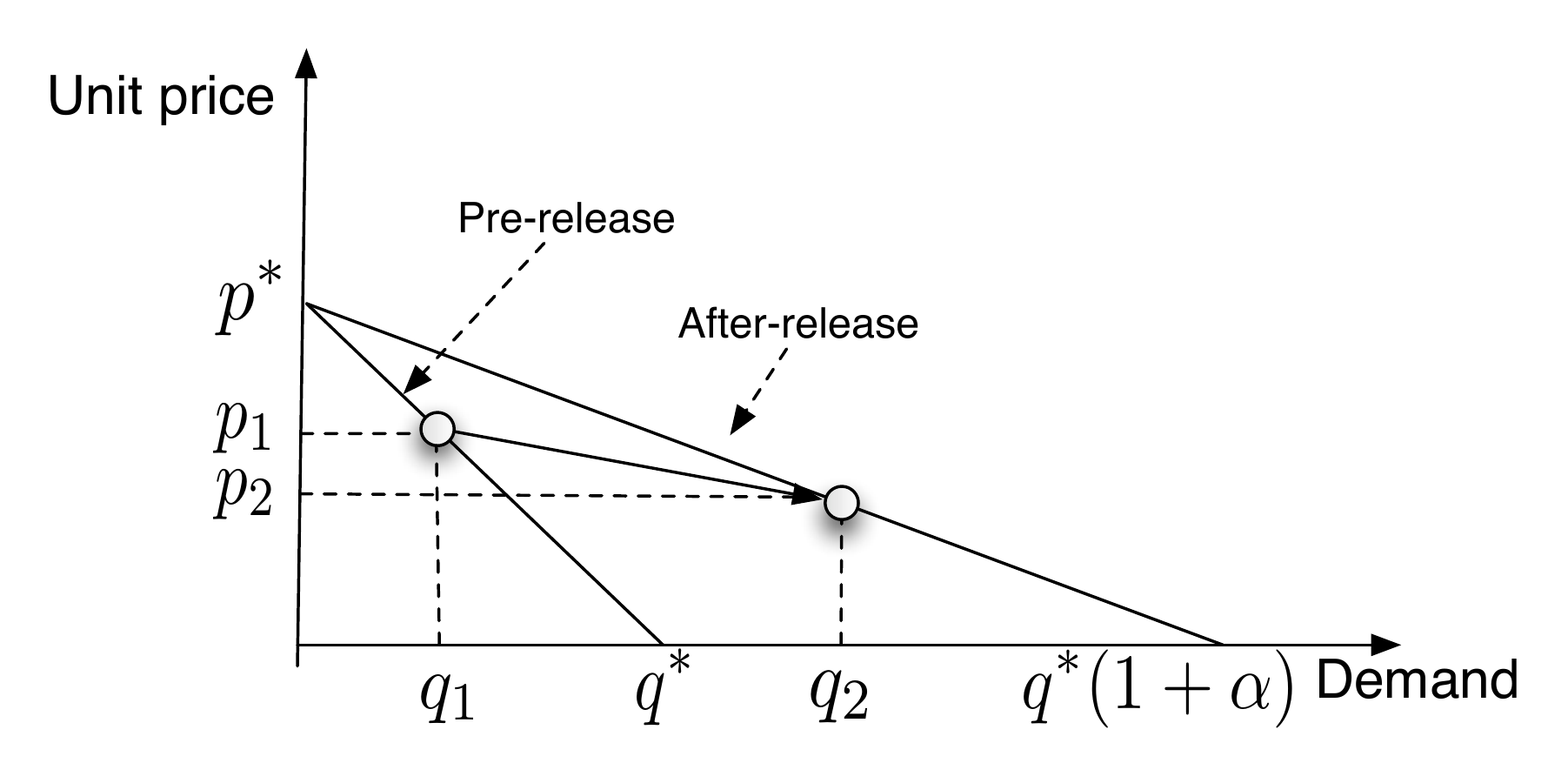}
	\caption{The demand curve before and after the release of personal data} 
	\label{fig:Domanda}
\end{center}
\end{figure}  

If we assume the willingness-to-pay to stay unchanged and the demand curve to be linear, the change in the demand curve is equivalent to a translation of the maximum amount of service, as illustrated in Figure~\ref{fig:Domanda}. When the customer releases the personal information contained in the set $I_{i}$ both the marginal demand (i.e., the increase in demand for a decreasing unit change in price) and the maximum consumption increase by the factor $(1+\alpha_{i})$, where $\alpha_{i} >0$ is the marginal demand factor. Since $I_{i}\subset I_{i+1}$ we can safely assume that $\alpha_{i} < \alpha_{i+1}$, with $i=1,2,\ldots , N-1$.  The new line passes through the points $(0,p^{*})$ and $(q^{*}(1+\alpha_{i}),0)$, so that its equation is now
\begin{equation}
\label{Dcurve}
\frac{q}{q^{*}(1+\alpha_{i})}+\frac{p}{p^{*}}=1.
\end{equation}

The movement of the customer on the new demand curve is a key point in determining its viability, and therefore the interest in releasing personal information. We have to make a distinction between customers whose consumption is time-constrained and customers who are not. 

If a user has a limited amount of time to use the service, we can assume that the extension of the allowed range of consumption is of no interest to it: time-constrained customers will consume the same quantity of service regardless of the new offer, so that $q_{2}=q_{1}$, and will not find convenient to release personal data to get a new offer.

Customers that are not time-constrained can in principle move to any position on the new demand curve. However, we expect that their moves, as well as the provider's offer, aim at increasing the respective satisfaction. In turn, that means that customers will change their consumption so as to increase their surpus, while a service provider will be willing to put forward a new offer if that leads to a revenue increase. Both assumptions pose conditions on the set of working points on the new demand curve; their intersection defines the set of valid working points.

We examine first the condition that the customer's surplus must increase. By moving to the new demand curve, the customers changes its surplus from $S_{c1}$ to $S_{c2}$:
\begin{equation}
\label{Vincuser}
S_{c1}=q_{1}\frac{p^{*}-p_{1}}{2} \longrightarrow S_{c2}=q_{2}\frac{p^{*}-p_{2}}{2}.
\end{equation} 
Since $p^{*}-p_{1}=p^{*}q_{1}/q^{*}$ and $p^{*}-p_{2}=p^{*}q_{2}/[q^{*}(1+\alpha)]$, the customers' surplus increases iff the following condition is satisfied
\begin{equation}
S_{c2}>S_{c1} \Longrightarrow q_{2}>q_{1}\sqrt{1+\alpha}.
\end{equation}

We now deal with the thrust that moves the service provider: getting more revenues. On the first demand curve, the revenues of the service provider, as given by the product of price and quantity, are
\begin{equation}
S_{p1}=p^{*}\left( 1-\frac{q_{1}}{q^{*}}\right)q_{1}.
\end{equation} 
After the release of personal information and the passage to the new demand curve, the revenues become
\begin{equation}
S_{p2}=p^{*}\left[ 1-\frac{q_{2}}{(1+\alpha)q^{*}}\right]q_{2}.
\end{equation}
The revenues increase if $S_{p2}>S_{p1}$, which in turn leads to the following quadratic inequality
\begin{equation}
q_{2}^{2}-(1+\alpha)q^{*}q_{2}+(1+\alpha)(q^{*}-q_{1})q_{1}<0.
\end{equation}
By solving it, we obtain that the passage to the new demand curve leads to increased revenues for the service provider if the new demand satisfies the inequality
\begin{equation}
\label{Vincprov}
(1+\alpha)q^{*}\frac{1-\sqrt{1-4\frac{q_{1}}{q^{*}}\frac{1-q_{1}/q^{*}}{1+\alpha}}}{2}<q_{2}<(1+\alpha)q^{*}\frac{1+\sqrt{1-4\frac{q_{1}}{q^{*}}\frac{1-q_{1}/q^{*}}{1+\alpha}}}{2}.
\end{equation}

By combining the two constraints represented by the inequalities (\ref{Vincuser}) and (\ref{Vincprov}), a necessary condition for the release of private information to be beneficial both to the customer and to the service provider is that the new demand satisfies the inequality
\begin{equation}
\label{Vinctot}
q_{1}\sqrt{1+\alpha}<q_{2}<(1+\alpha)q^{*}\frac{1+\sqrt{1-4\frac{q_{1}}{q^{*}}\frac{1-q_{1}/q^{*}}{1+\alpha}}}{2}.
\end{equation} 

We notice that, if the service provider maintains its unit price, the customer, acting as a price taker, changes its demand to
\begin{equation}
q_{2}=q_{1}(1+\alpha),
\end{equation}
which is surely within the acceptable region defined by the inequality (\ref{Vinctot}). In fact, the constraint on the customer's surplus is satisfied, since $q_{1}(1+\alpha)>q_{1}\sqrt{1+\alpha}$. And the provider's revenues increase, since the unit price has been kept constant and the demand has increased.

The width of the valid region represented by the inequality (\ref{Vinctot}) shows that the provider may increase its revenues even if the unit price is lowered.
 
But the disclosure of personal information comes with a cost: its potential malicious use by third parties. That misuse may turn into an economic loss for the customer. This may be both a direct loss deriving from the identity theft and an indirect one due, e.g., to the expenses incurred to recover deleted data, the loss of reputation, legal expenses, or the inability to use networking services. Each release of information $I_{i}$ is then associated to a potential money loss $l_{i}$, with $i=1,2,\ldots , N$. In the following we assume that all quantities (e.g., the losses and the quantity of service) are referred to a unit period of time, be it a month or a year. Again we can safely assume that the potential loss grows with the size of the information set, i.e., $l_{i}< l_{i+1}$. Since we have at the same time a positive effect (increase of services) and a negative one (potential money loss) linked to the same root cause (the disclosure of personal information), we can draw a relationship between the two effects through the common cause. In this paper we assume that relationship to be expressed through a power law:
\begin{equation}
\label{Info}
\frac{\alpha_{i}}{\alpha_{N}}=\left( \frac{l_{i}}{l_{N}}\right)^{\nu}.
\end{equation}
In addition to its well known property of scale invariance and its appearance in a number of contexts (see, e.g., \cite{Newman}\cite{Roberts}), the choice of a power law allows us to describe a variety of behaviours by acting on the single parameter $\nu$, which we call the privacy parameter. 
If $\nu <1$, the customer releases its information starting with the most potentially damaging, and the additional risk associated to further releases is a decreasing function of the information released. In particular, if $\nu \ll 1$ (i.e., the service provider is privacy-friendly), the customer gains a large benefit (i.e., a large extension of the maximum quantity of services) even for small pieces of the information released (i.e., small potential losses). When $\nu =1$, we have instead a linear relationship between the information released and the associated economical loss. The case $\nu > 1$ models instead the situation where the customer releases information starting with the least sensitive one. Here we don't support strongly any specific value forthe privacy parameter. However, we could set its value by exploiting a single instance of eq. (\ref{Info}), i.e., considering the fraction of the maximum potential loss $l_{i}/l_{N}$ corresponding to a given fraction of the benefit obtained $\alpha_{i}/\alpha_{N}$. For example, we might apply the Pareto principle, which states that roughly 80\% of the effects come from 20\% of the causes. That 80/20 relationship has been observed in many cases in the context of information security (see, e.g., \cite{Chen} \cite{Chung}). If we adopt that principle that the 80\% of the benefit is obtained with the 20\% of the effort, and assume $\alpha_{i}/\alpha_{N}=0.8$ and $l_{i}/l_{N}=0.2$ in Eq. (\ref{Info}), we obtain $\nu \simeq 0.138647$. Empirical data that shed some light on the nature of the information release law embodied by Equation (\ref{Info}) are provided in \cite{Riederer2011}: a transactional privacy mechanism has been proposed, whereby customers choose to release access to their personal information (namely, their browsing behaviour) in exchange for money. In the transactional privacy mechanism, customers are led to release their personal information starting with the least sensitive ones, since this strategy optimizes their revenues through that privacy selling mechanism.

\section{Data vulnerability}
\label{Pdb}
In Section \ref{Demand} we have shown that the extension of service provisioning (i.e., the increase in the maximum consumption by the factor $1+\alpha_{i}$) is related to the potential loss $l_{i}$ that the customer incurs when it releases the personal information $I_{i}$ and a data breach occurs. Such loss is uncertain; we need to associate that value with the probability that a data breach occurs. In this section we provide a model for such occurrence.

We consider that a data breach may take place because of deficiencies on either of the two sides of the customer-service provider relationship. The data theft may be due either to an attack on the service provider's information system or to the customer's data repository (e.g., its computer). We assume that the failures on the two sides are independent of each other, and that a data breach takes place as either of the two sides fail. Under these hypotheses, a suitable model for the overall data breach phenomenon is the classical series combination of two systems that we can borrow from the reliability field (see Ch. 3.2 in \cite{Gnedenko}). The data breach probability $\pi$ is then related to the individual data breach probabilities of the two sides $\pi_{\textrm{s}}$ (service provider) and $\pi_{\textrm{c}}$ (customer) by the formula
\begin{equation}
\label{Pdb1}
\pi=\pi_{\textrm{s}} + \pi_{\textrm{c}} -  \pi_{\textrm{s}}\cdot \pi_{\textrm{c}}.
\end{equation}

As to the vulnerability on the customer's side we consider that the probability of data breach is a growing function of the amount of personal information that the customer has divulged, since that determines its level of exposure to security threats. We assume a simple power law function to hold, and, by exploiting the relationship between information released and potential loss extablished in Section \ref{Demand}, we obtain the following function:
\begin{equation}
\label{Pdb2}
\pi_{\textrm{c}}=\pi_{\textrm{c}}^{*}\left( \frac{l_{i}}{l_{N}} \right)^{\theta},
\end{equation}
where $\pi_{\textrm{c}}^{*}$ is the probability of breach corresponding to the maximum release of information. The parameter $\theta \in (0,1)$ describes the balance between the probability of breach and the quantity of personal information released (for which the economical loss represents a proxy): if $\theta \ll 1$ (reckless customer) the probability of data breach is close to its maximum even for the smallest pieces of released information; if $\theta \simeq 1$ (privacy-aware customer) the customer has to release a substantial amount of information before it suffers a significant probability of data breach. We call $\theta$ the security parameter. We can set a value for it, again by resorting to the Pareto principle invoked for the privacy parameter.

For data breaches on the service provider's side a model has been proposed by Gordon and Loeb \cite{Gordon} to describe the dependence of the data breach probability on the investments in security carried out by the service provider. However, in this context we prefer not to explicitly account for the dependence of the data breach probability on that, or other determinants. In fact, the inclusion of specific variables affecting the data breach probability would leave the customer with the problem of determining their value to derive its optimal information disclosure strategy: in the Gordon-Loeb model the customer would have to guess the amount of security investment. Instead, we consider the system-related data breach probability $\pi_{\textrm{s}}$ as a parameter. We assume that the customer may estimate $\pi_{\textrm{s}}$ through a number of sources, and certainly more easily than the value of investments or any other information private to the service provider. 

\section{Maximization of customer's surplus}
\label{Optim}
In Sections \ref{Demand} and \ref{Pdb} we have provided models respectively for the price paid by the customer and for the data vulnerability, and we have also linked those quantities to the potential loss suffered by the customer through the disclosure of its personal information. In this section we use those relationships to determine the net surplus for the customer and the best information disclosure strategy.

The net surplus for the customer is given by the algebraic sum of two terms. The positive component is the surplus obtained by paying a price for the service lower than its willingness-to-pay. This surplus grows with the quantity of service that the customer receives, which in turn grows with the personal information that the customer reveals. Hence, the positive term grows with the amount of disclosed information (and with the related potential loss). On the other hand, the customer suffers a potential loss due to the same information disclosure; this is the negative term that again grows with the amount of disclosed information. We therefore expect to find a trade-off value for that amount of disclosed information that maximizes the net surplus. That represents the strategy the customer has to follow when revealing its private information to the service provider. In that strategy the customer acts as a price taker: the price is set by the service provider and the customers responds with a consumption dictated by the after-release demand curve. The latter depends however on the quantity of personal information released by the customer: different degrees of information correspond to different slopes of the demand curve shown in Figure~\ref{fig:Domanda}. The quantity of personal information released is then the leverage employed by the customer to maximize its surplus. Hereafter we provide the expression of the net surplus. 

It is to be noted that we do not assume any specific relationship between the information released by the customer and the potential loss for it; we just assume that there is one, and that the customer knows it (we expect that relationship to be different for each customer). In the following, the surplus is maximized by considering the loss as a leverage, since it actually represents a proxy for the information released. If the customer is able to identify the value of the loss that maximizes its surplus, at the same time it can identify the associated amount of information that maximizes the customer's surplus.

When the customer reveals the set of information $I_{i}$, its net surplus is given by the difference of the two terms mentioned above. If we indicate the price associated to the generic quantity $y$ as $p(y)$, we have
\begin{equation}
\label{Nsur}
S_{\textrm{c}}=\int_{0}^{q}\left( p(y)-p\right)dy-\pi [l_{i}],
\end{equation} 
where the first term embodies the surplus deriving from paying a price lower than the willingness-to-pay, integrated over the quantity of service received by the customer. 
We can now recall that: a) the unit price $p$ and the quantity of service $q$ are related through the demand curve (\ref{Dcurve}), namely $q/[q^{*}(1+\alpha_{i})]=1-p/p^{*}$; b) the maximum amount of service is related to the potential loss through the power law (\ref{Info}); c) the data breach probability is represented by expr. (\ref{Pdb1}) and (\ref{Pdb2}). With that additional information the net surplus can be expressed in the following form
\begin{equation}
\label{Nsur2}
\begin{split}
S_{\textrm{c}}&=\int_{0}^{q}\left\{ p^{*}\left[ 1-\frac{q}{q^{*}(1+\alpha_{i})}\right]-p\right\}-\pi l_{i}\\
&=(p^{*}-p)q-\frac{p^{*}}{q^{*}(1+\alpha_{i})}\frac{q^{2}}{2}-\pi l_{i}\\
&=p^{*}\left( 1-\frac{p}{p^{*}}\right)q-\frac{1}{2}p^{*}\left( 1-\frac{p}{p^{*}}\right)q-\pi l_{i}\\
&=\frac{1}{2}p^{*}\left( 1-\frac{p}{p^{*}}\right)q-\pi l_{i}\\
&=\frac{1}{2}p^{*}q^{*}(1+\alpha_{i})\left( 1-\frac{p}{p^{*}}\right)^{2}-\pi l_{i}\\
&=\frac{p^{*}q^{*}}{2}\left[ 1+\alpha_{N}\left( \frac{l_{i}}{l_{N}}\right)^{\nu}\right]\left( 1-\frac{p}{p^{*}}\right)^{2}-\left[ \pi_{\textrm{s}}+\pi_{\textrm{c}}^{*}(1-\pi_{\textrm{s}})\left(\frac{l_{i}}{l_{N}}\right)^{\theta}\right]l_{i}.
\end{split}
\end{equation}

If we stick to the discrete framework we have adopted in Section \ref{Demand}, the optimal amount of information to be disclosed by the customer is the set $I_{i}$ associated to the loss
\begin{equation}
\hat{l}=\underset{l_{i}}{\mathrm{argmax}} \quad S_{\textrm{c}}.
\end{equation}
However, we can gain a substantial insight into the trade-off problem, with a little loss in accuracy, if we now treat the potential loss as a continuous quantity, i.e., by replacing $l$ for $l_{i}$ in expr. (\ref{Nsur2}). The optimal loss can now be derived as
\begin{equation}
\hat{l}=l \quad : \quad \frac{\partial S_{\textrm{c}}}{\partial l}=0.
\end{equation}
If we derive expr. (\ref{Nsur2}), we get
\begin{equation}
\frac{\partial S_{\textrm{c}}}{\partial l}=\frac{q^{*}p^{*}\nu}{2}\frac{\alpha_{N}}{l_{N}}\left( 1-\frac{p}{p^{*}}\right)^{2}\left( \frac{l}{l_{N}}\right)^{\nu -1}-\pi_{\textrm{s}}-\pi_{\textrm{c}}^{*}(1-\pi_{\textrm{s}})(\theta+1)\left( \frac{l}{l_{N}}\right)^{\theta}.
\end{equation}
By equating that derivative to zero we get the optimal loss as the solution of the equation
\begin{equation}
\label{Solu}
\frac{q^{*}p^{*}\nu}{2}\frac{\alpha_{N}}{l_{N}}\left( 1-\frac{p}{p^{*}}\right)^{2}\left( \frac{l}{l_{N}}\right)^{\nu -1}-\pi_{\textrm{s}}+\pi_{\textrm{c}}^{*}(1-\pi_{\textrm{s}})(\theta+1)\left( \frac{l}{l_{N}}\right)^{\theta}=0,
\end{equation}
which we can rewrite in the synthetic form
\begin{equation}
\label{Eqopt}
Al^{\nu-1}-\pi_{\textrm{s}}-Bl^{\theta}=0,
\end{equation}
where $A$ and $B$ are the following positive quantities
\begin{equation}
\label{defAB}
\begin{split}
A&=\frac{q^{*}p^{*}\nu}{2}\frac{\alpha_{N}}{l_{N}^{\nu}}\left( 1-\frac{p}{p^{*}}\right)^{2}\\
B&=(1-\pi_{\textrm{s}})\pi_{\textrm{c}}^{*}\frac{\theta+1}{l_{N}^{\theta}}.
\end{split}
\end{equation}
We note that the derivative $\partial S_{\textrm{c}} / \partial l$ is made of three terms, two of which are powers of $l$, while the third one is a constant. That derivative is then a continuous function of $l$ over the interval $(0,\infty)$.

Solving the decision equation (\ref{Eqopt}) provides us with the optimal amount of information (for which the economical loss is proxy) that the customer can release. However, since we are interested in solutions that are at the same time positive and within the interval $[0,l_{N}]$, identified as the range of potential money losses embodied in expr. (\ref{Info}), we introduce the following definitions:
\begin{definition}[Legal solution]
A legal solution of the trade-off problem is any value $\hat{l}\geq 0$ for which the customer's surplus reaches a maximum over the interval $(0,\infty)$.
\end{definition}

\begin{definition}[Feasible solution]
A feasible solution of the trade-off problem is any value $l^{*} \in [0,l_{N}]$ such that $S_{\textrm{c}}(l^{*})\ge S_{\textrm{c}}(l)$,  $\forall l\in [0,l_{N}]$.
\end{definition}

Then, we are actually looking for a feasible solution of the trade-off problem. We can state the following theorem.
\begin{thm}[Trade-off solution]
If $\nu<1$ a feasible solution of the trade-off problem exists and is unique.\\
If $\nu>1$ a feasible solution of the trade-off problem exists and is unique if the following condition holds\\
\[
l_{N}<\alpha_{N}\frac{q^{*}p^{*}\nu}{2}\left( 1-\frac{p}{p^{*}}\right)^{2}\frac{1}{\pi_{\textrm{s}}+(1-\pi_{\textrm{s}})\pi_{\textrm{c}}^{*}(1+\theta)}
\]\\
If $\nu=1$ a feasible solution of the trade-off problem exists and is unique iff the following condition holds\\
\[
\frac{q^{*}p^{*}}{2}\left( 1-\frac{p}{p^{*}}\right)^{2}\frac{\alpha_{N}}{\pi_{\textrm{s}}+(1-\pi_{\textrm{s}})\pi_{\textrm{c}}^{*}(1+\theta)}<l_{N}<\frac{q^{*}p^{*}}{2}\left( 1-\frac{p}{p^{*}}\right)^{2}\frac{\alpha_{N}}{\pi_{\textrm{s}}}
\]
\end{thm}
\begin{proof} 
We consider the decision equation (\ref{Eqopt}), whose solution provides us with an extremal point of the customer's surplus. We deal separately with the three cases, where $\nu<1$, $\nu>1$, or $\nu =1$.\\
When $\nu <1$ it is convenient to rewrite the derivative of the customer's surplus as follows:
\begin{equation}
\frac{\partial S_{\textrm{c}}}{\partial l}=\frac{A}{l^{1-\nu}}-Bl^{\theta}-\pi_{\textrm{s}}.
\end{equation}
The behaviour of such derivative at either end of the interval over which we look for a legal solution is easily obtained as follows:
\begin{equation}
\begin{split}
\lim_{l \to 0^{+}} \frac{\partial S_{\textrm{c}}}{\partial l} &= +\infty,\\
\lim_{l \to +\infty} \frac{\partial S_{\textrm{c}}}{\partial l} &= -\infty .
\end{split}
\end{equation}
In addition the second derivative is
\begin{equation}
\label{secder}
\frac{\partial^{2} S_{\textrm{c}}}{\partial l^{2}}=A(\nu -1)l^{\nu-2}-B\theta l^{\theta -1}<0, \quad l\in (0,+\infty),
\end{equation}
so that the first derivative $\partial S_{\textrm{c}}/\partial l$ is a monotonic decreasing function.

We can find a closed interval $[l_{\textrm{l}},l_{\textrm{u}}]$ such that $\partial S_{\textrm{c}}/\partial l >0$ in $l=l_{\textrm{l}}$ and $\partial S_{\textrm{c}}/\partial l <0$ in $l=l_{\textrm{u}}$. We can prove this statement by construction. The upper bound $l_{\textrm{u}}$ can be found by solving the expression
\begin{equation}
\frac{A}{l_{\textrm{u}}^{1-\nu}}=Bl_{u}^{\theta} \longrightarrow l_{\textrm{u}}=\left( \frac{A}{B}\right)^{\frac{1}{\theta +1-\nu}},
\end{equation}
and noting that 
\begin{equation}
\frac{\partial S_{\textrm{c}}}{\partial l}_{\vert {l=l_{\textrm{u}}}}=\frac{A}{l_{\textrm{u}}^{1-\nu}}-Bl_{\textrm{u}}^{\theta}-\pi_{\textrm{s}}<\frac{A}{l_{\textrm{u}}^{1-\nu}}-Bl_{\textrm{u}}^{\theta}=0.
\end{equation}
As to the lower bound, we can set $l_{\textrm{l}}$ so that $Al_{\textrm{l}}^{\nu -1}=Al_{\textrm{u}}^{\nu -1}+\pi_{\textrm{s}}$. In $l=l_{\textrm{l}}$ the derivative of the customer's surplus is then
\begin{equation}
\frac{\partial S_{\textrm{c}}}{\partial l}_{\vert {l=l_{\textrm{l}}}} = Al_{\textrm{l}}^{\nu -1}-Bl_{\textrm{l}}^{\theta}-\pi_{\textrm{s}}=Al_{\textrm{u}}^{\nu -1}+\pi_{\textrm{s}}-Bl_{\textrm{l}}^{\theta}-\pi_{\textrm{s}}=Bl_{\textrm{u}}^{\theta}-Bl_{\textrm{l}}^{\theta}>0.
\end{equation} 
By Bolzano's theorem we can then conclude that there exists a value $\hat{l}\in (l_{\textrm{l}},l_{\textrm{u}}) : \frac{\partial S_{\textrm{c}}}{\partial l} \vert_{l=\hat{l}}=0$. Such value is an extremal point for the customer's surplus. Since the second derivative is strictly positive over the same interval, the customer's surplus reaches its maximum in $l=\hat{l}$. In addition, since the first derivative is a monotonic function, its inverse is unique; hence there is a single value $\hat{l}$.

We have so far a legal solution of the trade-off problem. In order to find a feasible solution we note that if $\hat{l}\in (0,l_{N})$, then the legal solution is also feasible, i.e., $l^{*}=\hat{l}$. Instead, if $\hat{l}\ge l_{N}$, then the feasible solution is $l^{*}=l_{N}$, since in that case $\partial S_{\textrm{c}}/\partial l >0$ when $l<l_{N}<l^{*}$ and the customer's surplus is a growing function over the interval $(0,l_{N}]$.

We consider now the case $\nu>1$ and further subdivide into the two subcases $1<\nu < 1+\theta$ and $\nu > 1+\theta$, which we label subcases a) and b) respectively.\\
In both subcases we have
\begin{equation}
\lim_{l \to 0^{+}} \frac{\partial S_{\textrm{c}}}{\partial l}=-\pi_{\textrm{s}}<0.
\end{equation}
But the behaviour of the partial derivative at the other end of the interval is different in the two subcases
\begin{equation}
\lim_{l \to +\infty} \frac{\partial S_{\textrm{c}}}{\partial l} = \left\{ \begin{array}{ll}
-\infty & \textrm{if $1<\nu < 1+\theta$}\\ +\infty & \textrm{if $\nu > 1+\theta$}
\end{array} \right.
\end{equation}

In subcase a) we have then a first derivative that takes the same sign at both ends of the interval $(0,\infty)$. We have a legal solution if that derivative takes the zero value within that interval. In order to understand what happens,we must look at the second derivative (\ref{secder}), which happens to get mixed signs. Namely, when $1<\nu < 1+\theta$, we have
\begin{equation}
\label{secdera}
\frac{\partial^{2} S_{\textrm{c}}}{\partial l^{2}} \left\{ \begin{array}{ll}
>0 & \textrm{if $l<\left[ \frac{A(\nu -1)}{B\theta}\right]^{1/(1+\theta - \nu)}$}\\ <0 & \textrm{if $l>\left[ \frac{A(\nu -1)}{B\theta}\right]^{1/(1+\theta -\nu)}$}
\end{array} \right.
\end{equation}
This means that the first derivative starts at $-\pi_{\textrm{s}}$, first increases, reaches a peak, and then falls without end. The peak is reached when $l=\left[ \frac{A(\nu -1)}{B\theta}\right]^{1/(1+\theta - \nu)}$. We can have the two different trends shown in Figure~\ref{fig:casea}, depending on the sign of the first derivative at its peak. If the peak is negative, we have no legal solution (and \textit{a fortiori} no feasible solution); if the peak is positive, we have two legal solutions, but 0, 1, or 2 feasible solutions. 
\begin{figure}[tp]
\begin{center}	
  \includegraphics[width=.7\columnwidth]{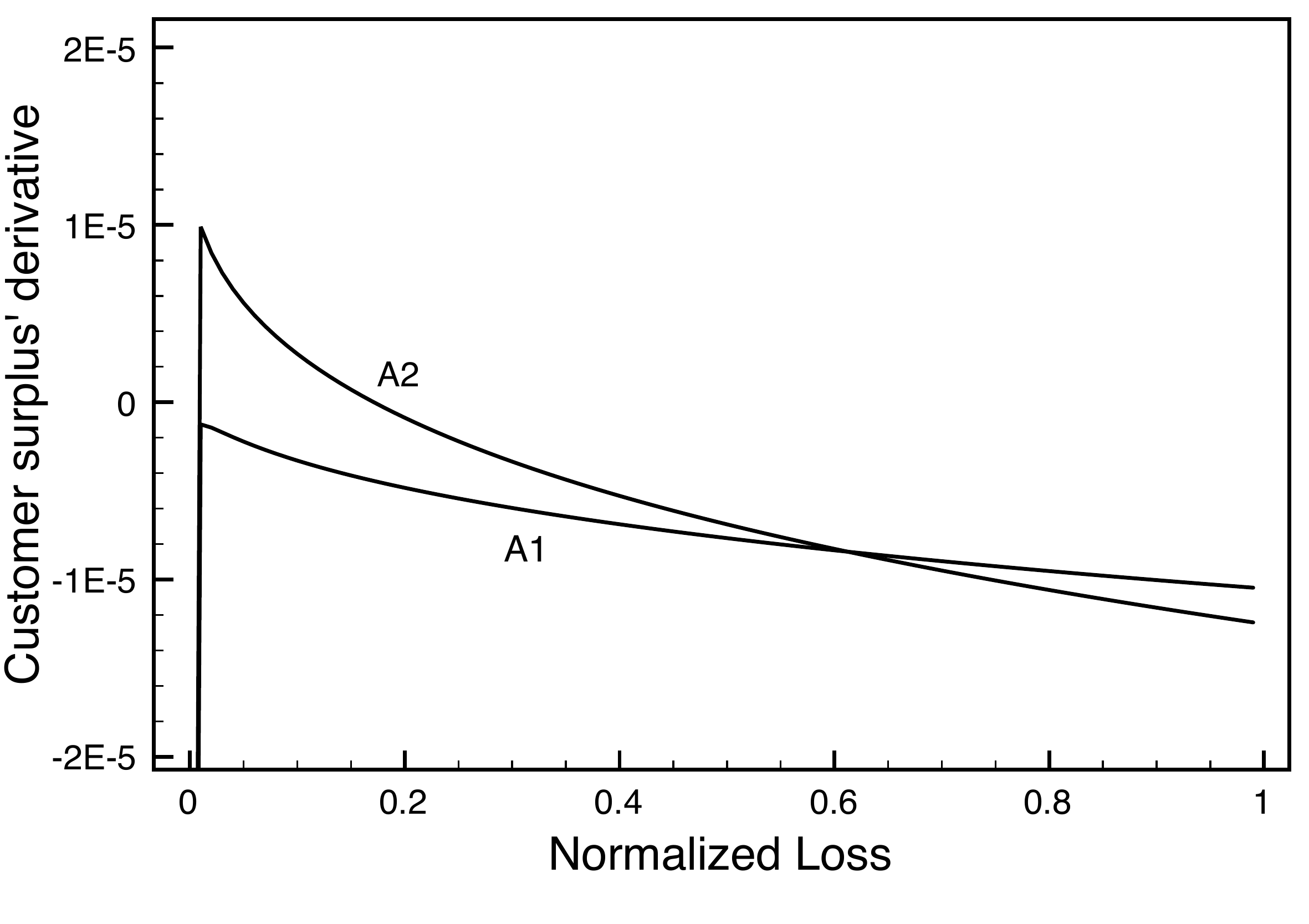}
	\caption{Alternative trends of the first derivative in subcase (a)}
	\label{fig:casea}
\end{center}
\end{figure}
In fact, even if the peak is positive, the number of feasible solutions depends on the position of the zeros. Since those zeros zeros represent two extremal points for the customer's surplus, we indicate them by the symbols $\hat{l}_{l}$ and $\hat{l}_{r}$ respectively, with $\hat{l}_{l}<\hat{l}_{r}$. We summarize the possible situations in Table \ref{table:nfeas}.
\begin{table}
\begin{tabular}{lc}
Condition & No. of feasible solutions\\
\hline
$\hat{l}_{l}>l_{N}$ & 0\\
$\hat{l}_{l}<l_{N}$ AND $\hat{l}_{r}>l_{N}$ & 1\\
$\hat{l}_{r}<l_{N}$ &  2\\
\end{tabular}
\label{table:nfeas}
\caption{Number of feasible solutions in subcase (a)}
\end{table}
We have then a single feasible solution just if the first derivative crosses the \textit{x}-axis downward at $l>l_{N}$. A sufficient condition for that to happen is that the the first derivative is still positive at $l_{N}$:
\begin{equation}
\frac{\partial S_{\textrm{c}}}{\partial l}_{\vert {l=l_{\textrm{N}}}}=\frac{\alpha_{N}}{l_{N}}\frac{q^{*}p^{*}\nu}{2}\left( 1-\frac{p}{p^{*}}\right)^{2}-\pi_{\textrm{s}}-(1-\pi_{\textrm{s}})\pi_{\textrm{c}}^{*}(1+\theta)>0,
\end{equation}
which can be reformulated as a condition on the maximum potential loss
\begin{equation}
\label{condsuff}
l_{N}<\frac{\alpha_{N}}{2}\frac{q^{*}p^{*}\nu}{\pi_{\textrm{s}}-(1-\pi_{\textrm{s}})\pi_{\textrm{c}}^{*}(1+\theta)}\left( 1-\frac{p}{p^{*}}\right)^{2}.
\end{equation}

In subcase (b), where $\nu > 1+\theta$, the first derivative takes opposite signs at either end of the $(0,+\infty)$. Since it is a continuous function, we are sure that it crosses the $x$-axis for some value of $l$. But it is not monotone, so that there might be more than one legal solution. This calls for a look at the second derivative (\ref{secder}), and we find that the situation is reversed with respect to Equation (\ref{secdera}) of case (a), namely
\begin{equation}
\label{secderb}
\frac{\partial^{2} S_{\textrm{c}}}{\partial l^{2}} \left\{ \begin{array}{ll}
<0 & \textrm{if $l<\left[ \frac{B\theta}{A(\nu -1)}\right]^{1/(\nu -1-\theta)}$}\\ >0 & \textrm{if $l>\left[ \frac{B\theta}{A(\nu -1)}\right]^{1/(\nu -1-\theta)}$}
\end{array} \right.
\end{equation}
According to (\ref{secderb}), the first derivative starts at $-\pi_{\textrm{s}}$, first decreases, and then grows without interruption, as depicted in Figure~\ref{fig:caseb}.
\begin{figure}[tp]
\begin{center}	
  \includegraphics[width=.7\columnwidth]{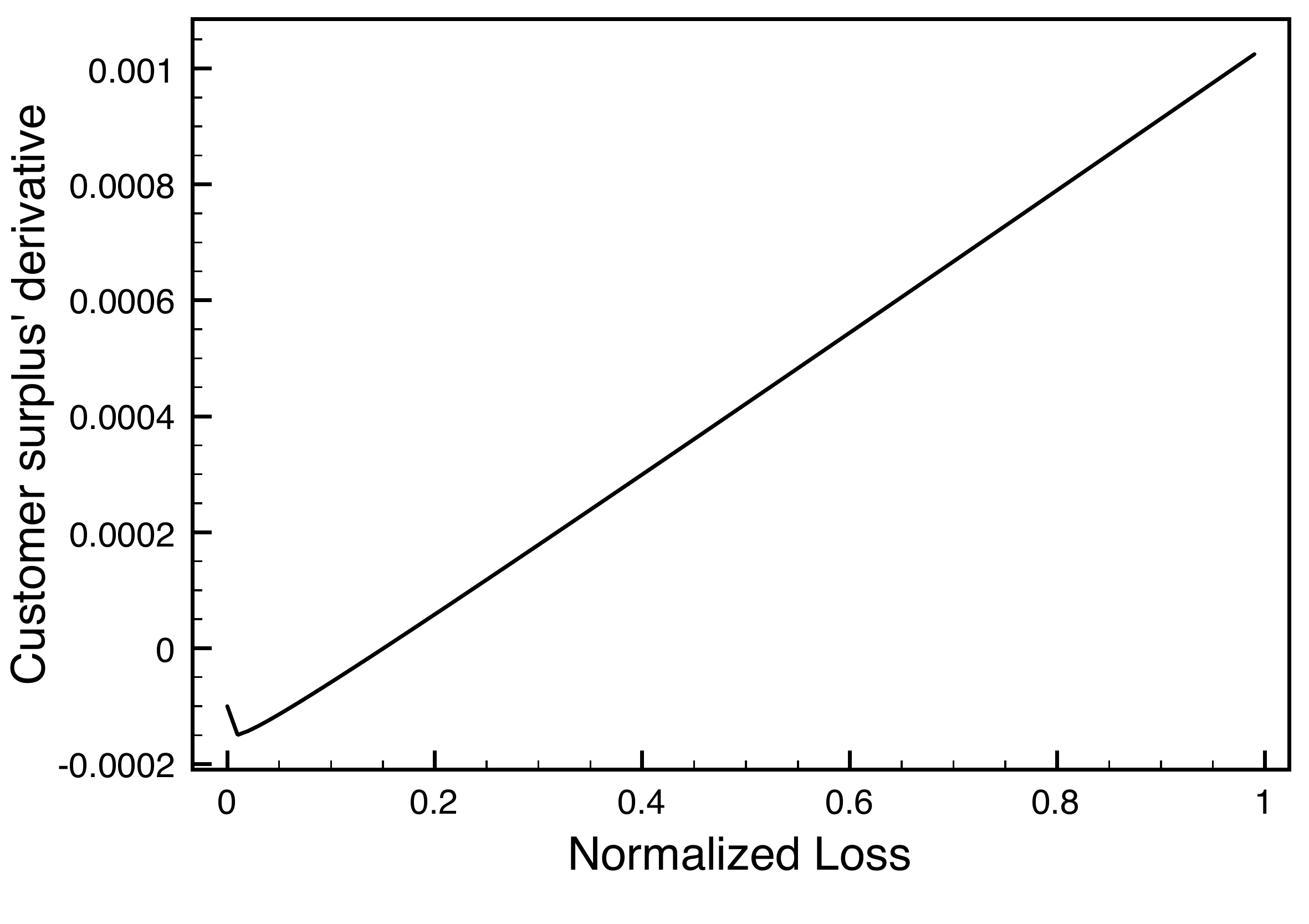}
	\caption{Trend of the first derivative in subcase (b)}
	\label{fig:caseb}
\end{center}
\end{figure}
This warrants a single legal solution. For it to be a single legal solution as well, we must evaluate the first derivative $\partial S_{\textrm{c}}/\partial l$  at $l=l_{N}$:  we have a single legal solution iff it is positive. This turns out to be the same condition as (\ref{condsuff}), which is then a sufficient condition for a single feasible solution both in subcases (a) and (b), i.e., when $\nu >1$.

Finally, when $nu = 1$, the decision equation (\ref{Eqopt}) reduces to
\begin{equation}
A-\pi_{\textrm{s}}-Bl^{\theta}=0 \rightarrow l^{\theta}=\frac{A-\pi_{\textrm{s}}}{B},
\end{equation}
which leads to the unique legal solution
\begin{equation}
l=\left[ \frac{A-\pi_{\textrm{s}}}{B} \right]^{1/\theta},
\end{equation}
if 
\begin{equation}
\label{caseqa}
A>\pi_{\textrm{s}}.
\end{equation}
This solution is also feasible if $l<l_{N}$, i.e.
\begin{equation}
\label{caseqb}
\frac{A-\pi_{\textrm{s}}}{B}<l_{N}^{\theta}.
\end{equation}
By introducing the definitions (\ref{defAB}) in the two previous conditions, we find that we have a unique legal solution of the decision equation (\ref{Eqopt}) iff
\begin{equation}
\frac{q^{*}p^{*}}{2}\left( 1-\frac{p}{p^{*}}\right)^{2}\frac{\alpha_{N}}{\pi_{\textrm{s}}+(1-\pi_{\textrm{s}})\pi_{\textrm{c}}^{*}(1+\theta)}<l_{N}<\frac{q^{*}p^{*}}{2}\left( 1-\frac{p}{p^{*}}\right)^{2}\frac{\alpha_{N}}{\pi_{\textrm{s}}}.
\end{equation}
The useful range obtained for the maximum loss $l_{N}$ is always nonzero since $(1-\pi_{\textrm{s}})\pi_{\textrm{c}}^{*}(1+\theta)>0$ 
\end{proof}

\section{Sensitivity of customer's feasible solution}
\label{Sens}
In Section \ref{Optim} we have shown that the trade-off problem has always a unique solution. The solution is represented by the maximum loss that the customer is willing to suffer to maximize its surplus. It is interesting to show the shape of the solution in a typical scenario such as that illustrated in Table \ref{table:T1}. In this section we examine the relationship between the optimal potential loss and the parameters appearing in the decision equation. We first consider the price imposed by the service provider as the single driving factor, and then analyze the sensitivity to all the other driving factors. Throughout this section, we consider the case $\nu < 1$.

We have stated in Section \ref{Optim} that the customer acts as a price taker. The price is then a relevant leverage that the service provider may use to drive the behaviour of the customer and eventually raise its tolerance towards taking more potential risks (i.e., increasing the maximum loss it tolerates). In order to examine the effect of price on the risk-taking attitude of the customer, we plot in Figure~\ref{fig:D1} the optimal potential loss for the customer when the other quantities take the values in Table \ref{table:T1}
\begin{table}
\begin{tabular}{lc}
Parameter & Value\\
\hline
Maximum quantity of service $q^{*}$ & 250\\
Willingness-to-pay $p^{*}$ & 1\\
Privacy parameter $\nu$ &  0.138647\\
security parameter $\theta$ & 0.138647\\
Maximum loss $l_{N}$ & 5000,10000\\
Marginal demand factor $\alpha_{N}$ & 20\%\\
System Data Breach Probability $\pi_{\textrm{s}}$ & $10^{-5}$\\
Maximum Customer Data Breach Probability $\pi_{\textrm{c}}^{*}$ & $10^{-4}$\\
\end{tabular}
\label{table:T1}
\caption{Values of parameters for the case study}
\end{table}

\begin{figure}[tp]
\begin{center}	
  \includegraphics[width=.7\columnwidth]{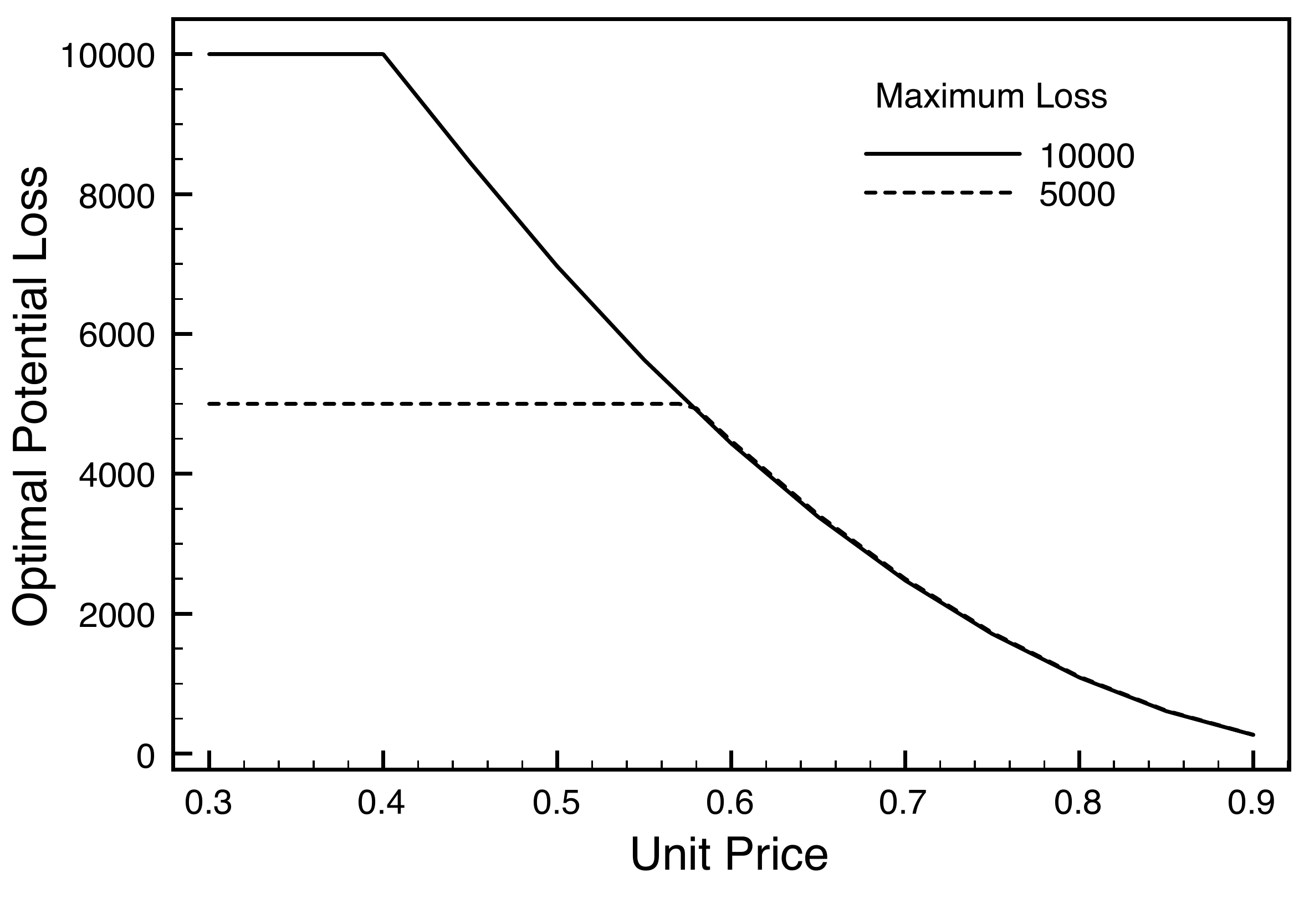}
	\caption{Impact of the unit price on the optimal amount of information released}
	\label{fig:D1}
\end{center}
\end{figure} 
Aside from the initial flat portion (due to the limitation on the maximum loss $l_{N}$) the optimal potential loss $l^{*}$ is a decreasing function of the price. We see that the maximum loss has no effect on the curve, excepting the cap that determines the initial flat portion. Hence, we can conclude that increasing the price makes the customer more careful in taking security-related risks. If the service provider wishes to relax the customer's attitude towards taking risks (in order to gather more personal information), it should then lower the price, though it is not necessary to go below the threshold corresponding to the saturation value (i.e., the maximum loss, equal to either 5000 or 10000 in Figure~\ref{fig:D1}).

On the other hand, the service provider could use the price for reasons other than driving the customer's behaviour towards taking risks, for example to maximize its profit. Though it can derive profits from exploiting the customer's personal information, we can even consider just the profits deriving from the customer's payment to gain some insight. Actually, for the same data  of Table \ref{table:T1} (with $l_{N}=10000$), we can see in Figure~\ref{fig:D2} that there is an optimal price maximizing the revenues. This optimal price is roughly 0.5 in this case, slightly larger than the maximum price needed to get all the personal information from the customer (which is roughly 0.4 in this case). Hence, by setting the price in this range, the service provider may strike a balance (and achieve the overall maximum profit) between making a profit by exploiting the customer's personal information or by extracting revenues from the customer's payments.
\begin{figure}[tp]
\begin{center}	
  \includegraphics[width=.7\columnwidth]{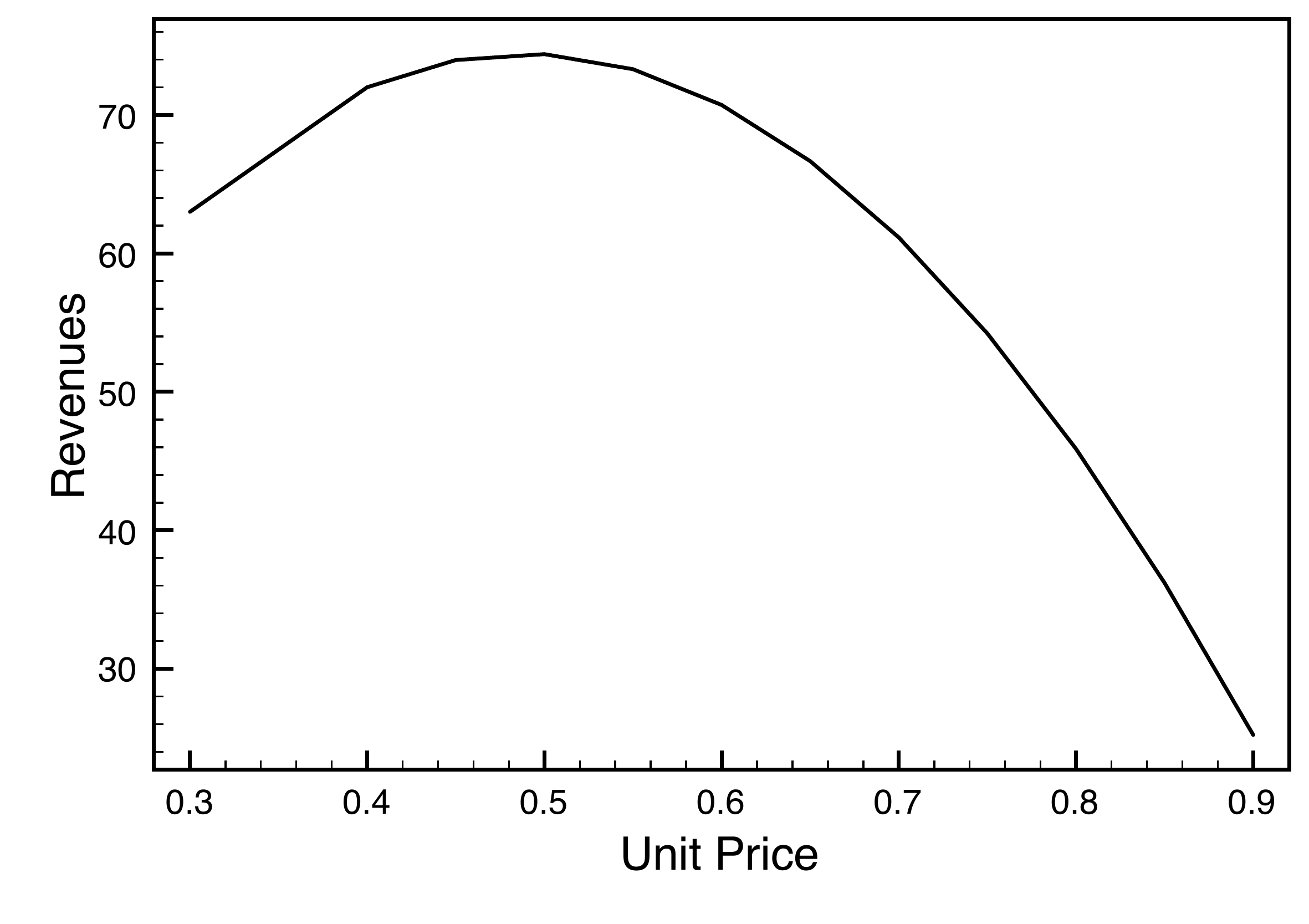}
	\caption{Impact of the unit price on the revenues collected by the service provider}
	\label{fig:D2}
\end{center}
\end{figure} 

In order to analyse the impact that each factor intervening in the decision equation has on the trade-off solution, we impose a change to each of the driving factors separately (only one at a time) and observe the resulting change in the optimal potential loss. Among the driving factors undergoing the same change that resulting in the largest change in the optimal potential loss can be considered as the most relevant. We take the scenario described in Table \ref{table:T2} as a reference. The optimal potential loss in that case is $l^{*}=3797$.

\begin{table}
\begin{tabular}{lc}
Parameter & Value\\
\hline
Maximum quantity of service $q^{*}$ & 250\\
Willingness-to-pay $p^{*}$ & 1\\
Price $p$ & 0.5\\
Privacy parameter $\nu$ &  0.138647\\
security parameter $\theta$ & 0.138647\\
Maximum loss $l_{N}$ & 10000\\
Marginal demand factor $\alpha_{N}$ & 20\%\\
System Data Breach Probability $\pi_{\textrm{s}}$ & $10^{-4}$\\
Maximum Customer Data Breach Probability $\pi_{\textrm{c}}^{*}$ & $10^{-4}$\\
\end{tabular}
\label{table:T2}
\caption{Reference values for the sensitivity analysis}
\end{table}
In order to perform a sensitivity analysis we divide the parameters into two groups: the first group is composed of the parameters possessing a unit of measure, namely the maximum quantity of service $q^{*}$, the willingness-to-pay $p^{*}$, the price $p$, and the maximum loss $l_{N}$;  the second group is composed of the dimensionless variables $\nu$, $\theta$, $\pi_{\textrm{s}}$, and $\pi_{\textrm{c}}^{*}$. For the two groups we adopt two different sensitivity measures, namely the elasticity for the first group and the quasi-elasticity for the second group.

We start with the parameters possessing a unit of measure. For them we define a discrete version of the elasticity measure. The elasticity is the ratio of the percent change in the output variable to the percent change in the driving factor, when small changes are imposed on the driving factor. In our case the output variable is the optimal potential loss $l^{*}$, but we consider significant, rather than small, changes in the driving factor. With reference to the generic  driving factor $x$ (which may be either $q^{*}$, $p^{*}$, $p$ or $l_{N}$), the discrete elasticity is
\begin{equation}
\label{Elast}
\varepsilon_{x}=\frac{\Delta l^{*}/l^{*}}{\Delta x/x},
\end{equation}
where $\Delta l^{*}$ and $\Delta x$ are respectively the discrete changes of the optimal potential loss and the driving factor. Here we impose $\Delta x /x= \pm 10\%$ in turn for each of the driving factors, observe the resulting changes in the optimal potential loss, and compute the elasticity through Definition (\ref{Elast}). The driving factors exhibiting larger values of elasticity have a greater influence on the optimal potential loss and hence on the customers' attitude towards divulging personal data.. In order to compare the three driving factor at hand, we draw a Tornado chart. The Tornado chart is a special type of bar chart, well used to perform sensitivity analyses \cite{Eschenbach}; in our case the bars represent the values of the elasticity, and we arrange the driving factors vertically, ordered so that the largest bar appears at the top of the chart, the second largest appears second from the top, and so on.  The Tornado chart for $q^{*}$, $p^{*}$, and $l_{N}$ is shown in Figure~\ref{fig:T1}, where the darker bars are employed to show the elasticity for positive increments of the driving factor, and the lighter bars for negative ones. We see that there are very large differences in the impact of each quantity. Both the elasticities $\varepsilon_{q^{*}}$ and $\varepsilon_{p^{*}}$ are positive; hence, when either the willingness-to-pay or the maximum quantity of service increase/decrease, the optimal potential loss increases/decreases as well. Since $\varepsilon_{q^{*}}\simeq 1$, changes in the maximum quantity of service translate to roughly the same percentage change in the optimal potential loss. Instead, the effect of changes in the willingness-to-pay is nearly trebled, and the effect of price changes is nearly doubled. On the other end of the scale, the maximum  loss has a negligible impact on the optimal potential loss (as long as we do not enter the saturation region, i.e., the flat portion of Figure~\ref{fig:D1}, where $l^{*}=l_{N}$, so that, as long as we keep staying on the flat portion, each change of the maximum loss reflects in an equal increase of the optimal potential loss, or, in other terms, the elasticity is unitary). Both the elasticity with respect to the price and to the maximum loss are negative: if we increase either the maximum loss or the price the optimal potential loss decreases.
\begin{figure}[tp]
\begin{center}	
  \includegraphics[width=.7\columnwidth]{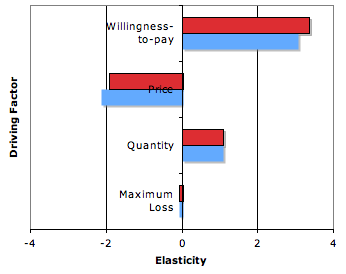}
	\caption{Elasticity for the maximum quantity of service, the willingness-to-pay, and the maximum loss}
	\label{fig:T1}
\end{center}
\end{figure}
We can then see that the maximum loss is, and by far, the least relevant factor in the customer's decision, while a greater role is played by the price and by the quantity of service. We can observe that a decrease in the price or an increase in the maximum quantity of service are benefits that the customer appreciates readily. Since the maximum loss is associated to the worst case, and its evaluation is often influenced by attention-grabbing catastrophic statements or press releases, we can then expect the effect of such statements to be quite small on the actual customer's behaviour. This conclusions aligns with the observed small impact of data breach notifications on identity thefts reported in \cite{Romanosky}. The major effects are then provided by immediate benefits (price and maximum quantity of services) rather than by prospective effects (the maximum loss that may occur in the future). In fact, it has been observed that customers tend to underestimate the present value of a future loss, e.g., by applying a hyperbolic discount rate and underinsure themselves against future risks \cite{Acquisti05}; such behaviour is then well predicted by our model. 

In addition to the degree of influence they exert on the customer's decision, the four quantities so far considered differ also as to the capability of the service provider to act on them. In fact, the willingness-to-pay and the maximum quantity of service depend mainly on macro-economic and social factors (e.g., inflation and salaries, service penetration level, social use of the service), on which the service provider has little or no influence. On the other hand, the service provider has a nearly total control on the price (excluding regulatory intervention and the market pressure). The price is therefore an easy leverage in the hands of the service provider. The service provider may exert a relevant influence on the maximum potential loss as well, for example by determining the level of anonymization applied to the personal data so to mitigate the losses in the case of a data breach.
 
For the dimensionless driving factors (i.e., the security parameter $\theta$, the privacy parameter $\nu$, and the data breach probabilities $\pi_{\textrm{s}}$ and $\pi_{\textrm{c}}^{*}$) we cannot use the elasticity. In fact, the popular interpretation of elasticity is that it represents the percentage change in the output variable upon a 1\% change in the driving factor, and the major reason why the elasticity is preferred to the derivative is that it is invariant to the arbitrary units of measurements of both variables. In the present case, however, the driving factor is either a probability or a bounded parameter, and its scale is not arbitrary since it ranges from 0 to 1. It is then usual to resort to the quasi-elasticity, where instead of two percentage changes we form the ratio of a percentage value and an absolute quantity (see, e.g., Section 2.1 of \cite{Cramer}). In our case, the quantity expressed as a percentage is the optimal potential loss. As for the elasticity, we recall that a positive sign of the quasi-elasticity means that the optimal potential loss changes in the same direction as the driving factor. For each dimensionless driving factor $x$, we define the following discrete version of the quasi-elasticity:
\begin{equation}
\label{Qelast}
\hat{\varepsilon}_{x}=\frac{\Delta l^{*}/l^{*}}{\Delta x}.
\end{equation} 
We deal first with the quasi-elasticity with respect to the security and privacy parameters. For both we examine what happens to the optimal potential loss when they take two values at either side of the reference value 0.138647, namely the values 0.1 and 0.2. We plot the resulting quasi-elasticities in Figure~\ref{fig:T2}, where we note that the privacy parameter exerts the larger influence by far. We add that the quasi-elasticity $\hat{\varepsilon}_{\nu}$ is positive, while $\hat{\varepsilon}_{\theta}$ is negative. We can conclude that the customer appears to be more sensitive to the profiling activities conducted by the service provider (embodied by the privacy parameter) than to its own self-protection attitude (embodied by the security parameter). If the service provider is privacy-friendly, the customer is willing to share more data (since the optimal potential loss grows). A privacy-friendly attitude by the provider spurs a positive market reaction, and  the opt-in approach (the approach requiring the express consent by the customer to be profiled or to receive information, merchandise, or messages from a marketer) is not a detrimental factor for the market \cite{A29,FTC2010}.  
\begin{figure}[tp]
\begin{center}	
  \includegraphics[width=.7\columnwidth]{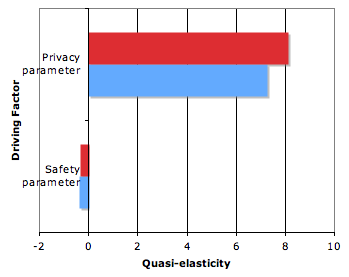}
	\caption{Quasi-Elasticity for the security and privacy parameters}
	\label{fig:T2}
\end{center}
\end{figure}

We plot again the Tornado chart for the quasi-elasticity with respect to the two data breach probabilities $\pi_{\textrm{s}}$ and $\pi_{\textrm{c}}^{*}$ in Figure~\ref{fig:T3}. For both probabilities the range considered is $(5\cdot 10^{-5},2\cdot 10^{-4})$, i.e., two octaves wide: the very large values (in the order of thousands) reported for the quasi-elasticity in this case are a consequence of the very small values of the data breach probabilities, which form the denominator in the ratio employed in the definition (\ref{Qelast}). In both cases the quasi-elasticity is negative: an increase in the data breach probability leads anyway the customer towards a more cautious behaviour. Here, the influence of the two driving factors is quite comparable. However, the influence of the customer's data breach probability is slightly larger.
\begin{figure}[tp]
\begin{center}	
  \includegraphics[width=.7\columnwidth]{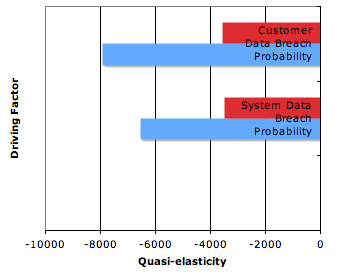}
	\caption{Quasi-Elasticity for the data breach probabilities}
	\label{fig:T3}
\end{center}
\end{figure}

\section{The case of a perfectly secure provider}
In Section \ref{Pdb} we have seen that both the customer and the provider have a role in determining the overall probability of a data breach. In the sample cases shown in Section \ref{Sens} we have considered the case where both parties are vulnerable. However, we may expect the provider's information system to be better protected than the individual customer against data thefts. It is then interesting to consider what happens when the data breach is entirely due to the customer's vulnerability. In this section we consider the limit case of the perfectly secure service provider and observe the effect on the customer's trade-off decision.

In order to examine the case of the perfectly secure provider, we go back to the decision equation (\ref{Solu}), and set the service provider's data breach probability $\pi_{\textrm{s}}=0$. We get the new decision equation
\begin{equation}
\label{Deq2}
\frac{q^{*}p^{*}\nu}{2}\frac{\alpha_{N}}{l_{N}}\left( 1-\frac{p}{p^{*}}\right)^{2}\left( \frac{l}{l_{N}}\right)^{\nu -1}+\pi_{\textrm{c}}^{*}(\theta+1)\left( \frac{l}{l_{N}}\right)^{\theta}=0,
\end{equation}
which can be solved for the optimal potential loss
\begin{equation}
\label{Solu2}
l^{*}=\left[\frac{q^{*}p^{*}\nu}{2}\frac{\alpha_{N}}{\pi_{\textrm{c}}^{*}(\theta+1)}l_{N}^{\theta - \nu}\left( 1-\frac{p}{p^{*}}\right)^{2}\right]^{1/(\theta - \nu +1)}.
\end{equation}
 
We expect the risky attitude of the customer to grow as an effect of the security assured by the provider. Actually, we can measure such effect by forming the ratio of the optimal loss values obtained when $\pi_{\textrm{s}}=0$ and $\pi_{\textrm{s}}\neq 0$, all the other parameters being equal. We may then define the Optimal Loss Ratio as follows
\begin{equation}
OLR=\frac{l^{*}\vert_{\pi_{\textrm{s}}=0}}{l^{*}\vert_{\pi_{\textrm{s}}\neq 0}}.
\end{equation}
If $OLR > 1$, then the customer is willing to accept a larger potential loss as a consequence of the perfect security offered by the service provider. We consider the sample case defined in Table \ref{table:T2} (but with $l_{N}=10000$) to see how the risk-taking attitude of the customer changes in the perfect provider scenario. We plot the $OLR$ ratio in Figure~\ref{fig:Olr}. We see that the customer is actually ready to increase its potential loss, even doubling it, and that the optimal loss ratio increases with the service price. The discontinuity exhibited by that ratio in the picture, roughly around $p=0.4$, is due to the optimal potential loss reaching its maximum value in the completely secure provider scenario while it is still growing in the unsecure provider case. 
\begin{figure}[tp]
\begin{center}	
  \includegraphics[width=.7\columnwidth]{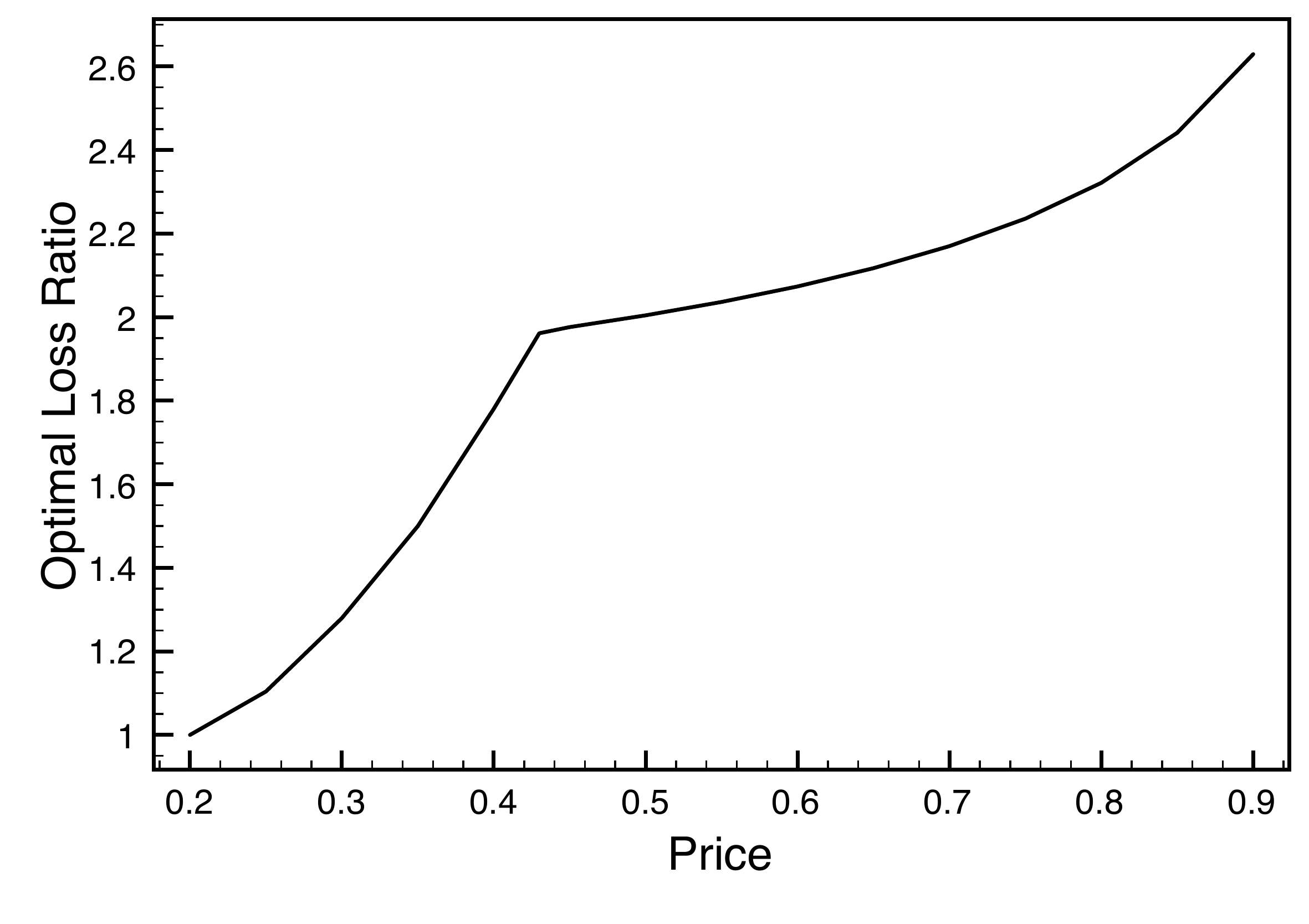}
	\caption{Increase in the optimal potential loss under complete provider's security}
	\label{fig:Olr}
\end{center}
\end{figure}

In addition to observing how the optimal potential loss moves when the service provider increases the price, we can examine the sensitivity of the trade-off solution to each of the variables appearing in the decision equation (\ref{Deq2}), i.e., the driving factors. In the case of the unsecure provider, where we had not a closed-form solution of the decision equation, we resorted to the numerically-evaluated discrete elasticity, defined for significant changes of the driving factor. Here, since we have the solution (\ref{Solu2}), we can use instead the elasticity
\begin{equation}
\varepsilon_{x}=\frac{x}{l^{*}}\frac{\partial l^{*}}{\partial x},
\end{equation}
where $x$ is again any of the driving factors. By deriving the expr. (\ref{Solu2}) with respect to each of the driving factors and replacing the trade-off solution and its first derivative in the definition of elasticity, we find the following expressions for the elasticity:
\begin{equation}
\begin{split}
\varepsilon_{q^{*}}&=\frac{1}{\theta - \nu +1},\\
\varepsilon_{p^{*}}&=\frac{1}{\theta - \nu +1}\frac{1+p/p^{*}}{1-p/p^{*}},\\
\varepsilon_{l_{N}}&=\frac{\theta - \nu}{\theta - \nu +1},\\
\varepsilon_{p}&=-\frac{1}{\theta - \nu +1}\frac{p/p^{*}}{1-p/p^{*}}.
\end{split}
\end{equation}
We note first that all the elasticities depend on the difference $\theta - \nu$, i.e., between the security parameter and the privacy parameter. Namely, the elasticities with respect to the maximum quantity and the maximum loss depend on that difference only. That difference also appears in the solution (\ref{Solu2}). The presence of the difference $\theta - \nu$ means that there is a sort of counterbalance between the two parameters: if the customer has set its optimal amount of shared personal data (i.e., its proxy optimal loss $l^{*}$) and the service provider gets looser in the treatment of personal data (the privacy parameter $\nu$ grows), the customer may take a counterbalancing move by increasing its self-protecting attitude, i.e., by increasing its security parameter $\theta$ so to keep the difference $\theta - \nu$ constant. Instead, the elasticities with respect to the price imposed by the service provider and the willingness-to-pay of the customer also depend on the $p/p^{*}$ ratio. When $-1\le \theta - \nu \le 1$ and $0\le p/p^{*} \le 1$, we can plot the elasticity as a function of $\theta - \nu$, with the  $p/p^{*}$ ratio as a parameter. We obtain the curves shown in Figure \ref{fig:elast} (where the $p/p^{*}$ ratio is indicated as PW). The elasticity with respect to the willingness-to-pay is always positive: easy spenders take larger risks with respect to low-spenders. On the other hand, the elasticity with respect to the price is always negative: an increase in price spurs a more restrained attitude towards divulging personal data. This observation confirms the conclusions reached in \cite{Taylor2004}, where it was shown that customers aware of the use of their personal data are less willing to be profiled (profiling takes place through recording the purchasing history of the customer) and refuse to buy at high prices. A similar attitude towards being profiled is also present in \cite{NaldiP10}, where the customer may refuse an offer by a service provider with the aim of getting a lower price. We see that in all cases, as the difference $\theta - \nu$ grows (we tend towards a combination of privacy-friendly provider and privacy-aware customer), the absolute value of the elasticity decays towards a limit value. In fact, when $\theta - \nu = 1$, we have $\varepsilon_{q^{*}}=\varepsilon_{l_{N}}=1/2<1$ (quite an anelastic response), so that the customer is relatively unwilling to increase its exposure to risk because of changes in the service offer by the service provider or in the maximum loss. In all cases we note that the ranking among the different elasticities observed for the case of unsecure provider in Section \ref{Sens} is preserved.

Similarly, for the parameters that do not possess a unit of measure, we use the quasi-elasticity
\begin{equation}
\hat{\varepsilon}_{x}=\frac{\partial l^{*} /l^{*}}{\partial x}.
\end{equation}
For the privacy and the security parameters and the customer's data breach probability we have respectively
\begin{equation}
\begin{split}
\hat{\varepsilon}_{\nu}&=\frac{1}{\theta - \nu +1}\left[ \frac{1}{\nu}+\ln \left(\frac{l^{*}}{l_{N}}\right)\right],\\
\hat{\varepsilon}_{\theta}&=-\frac{1}{\theta - \nu +1}\left[ \ln \left(\frac{l^{*}}{l_{N}}\right)-\frac{1}{\theta + 1}\right],\\
\hat{\varepsilon}_{\pi_{\textrm{c}}^{*}}&=-\frac{1}{\theta - \nu +1}\frac{1}{\pi_{\textrm{c}}^{*}}.
\end{split}
\end{equation}
We note that none of this set of quasi-elasticities is bounded. The quasi-elasticities with respect to the security and the privacy parameters may take both positive and negative values, while $\hat{\varepsilon}_{\pi_{\textrm{c}}^{*}}$ is always negative. Namely, we have $\hat{\varepsilon}_{\nu}>0$ iff $l^{*}/l_{N}>\exp (-1/\nu)$. If the privacy parameter is not very close to 1, that may happen nearly for the whole range of prices: for the reference scenario of Table \ref{table:T2}, the quasi-elasticity is positive if the optimal potential loss is larger than 7.13 (with $l_{N}=10000$), a condition met as long as $p<0.984 p^{*}$. Instead, as to the security parameter, we have $\hat{\varepsilon}_{\theta}>0$ iff $l^{*}/l_{N}<\exp (-1/(1+\theta))$. Again for the scenario of Table \ref{table:T2}, the quasi-elasticity is positive as long as $p>0.6305 p^{*}$. 

\begin{figure}[tp]
\begin{center}	
  \includegraphics[width=.7\columnwidth]{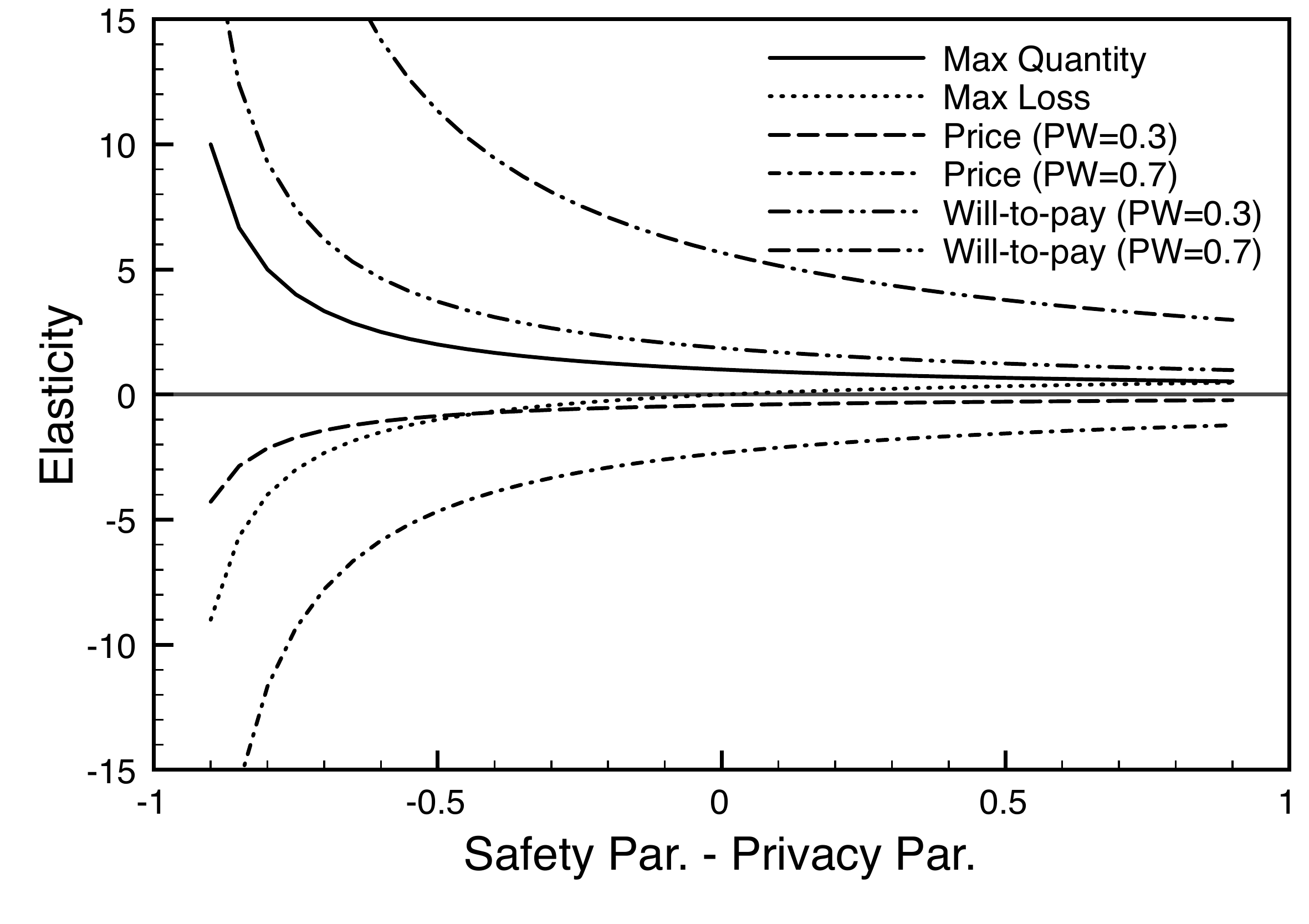}
	\caption{Increase in the optimal potential loss under complete provider's security}
	\label{fig:elast}
\end{center}
\end{figure}

\section{Conclusions}
We have formulated the decision to be taken by the customer on the amount of personal information to release as a trade-off problem, where the customer aims at maximizing its surplus represented by the algebraic sum of benefits and disadvantages. The benefit for the customer is the reduction of price (or, equivalently, the increase in the quantity of service obtained) achieved in return for the release of personal information. But the customer may suffer an economical loss due to the fraudulent use of its personal data. We have proved that the solution to the trade-off problem exists and is unique, providing the optimal potential loss associated to the amount of personal information the customer finds convenient to release. We have also performed a sensitivity analysis to identify the most influential driving factors, i.e., those parameters that have the largest influence on the trade-off solution. The major effects are provided by immediate benefits (price and maximum quantity of services) rather than by prospective negative effects (the maximum loss that may occur in the future). A major role is also played by the so-called privacy parameter, by which the service provider regulates the benefits released to the customers. For the special case of a perfectly secure provider (a data breach may occur just on the customer's side) we have provided a closed-form solution of the decision equaiton, and analytical expressions for the elasticity (or quasi-elasticity) with respect to all the variables involved. The results obtained for the perfectly secure provider are quite close to those obtained for the general case. The price and the willingness-to-pay play anyway a major role in the customer's decision. Easy spenders take larger risks with respect to low-spenders, but in general an increase in price spurs a more restrained attitude towards divulging personal data. The results may be used by the customer to determine its behaviour in response to the information disclosure requests coming from its service provider.


\begin{thebibliography}{10}

\bibitem{Milne}
G.R. Milne and M.E. Gordon.
\newblock Direct mail privacy-efficiency trade-offs within an implied social
  contract framework.
\newblock {\em Journal of Public Policy \& Marketing}, 12(2):206--215, 1993.

\bibitem{Romanosky}
S.~Romanosky, R.~Sharp, and A.~Acquisti.
\newblock Data breaches and identity theft: When is mandatory disclosure
  optimal ?
\newblock In {\em The Ninth Workshop on the Economics of Information Security
  (WEIS)}, Harvard University, USA, 7-8 June 2010.

\bibitem{Ponemon}
{The Ponemon Institute}.
\newblock {2009 Annual Study: Cost of a Data Breach}.
\newblock Technical report, January 2010.

\bibitem{LE2010}
{London Economics}.
\newblock {Study on the economic benefits of privacy-enhancing technologies
  (PETs). Final Report to the European Commission DG Justice, Freedom and
  Security}.
\newblock Technical report, July 2010.

\bibitem{Cavusoglu}
Huseyin Cavusoglu, Birendra Mishra, and Srinivasan Raghunathan.
\newblock A model for evaluating {IT} security investments.
\newblock {\em Commun. ACM}, 47:87--92, July 2004.

\bibitem{Gordon}
Lawrence~A. Gordon and Martin~P. Loeb.
\newblock The economics of information security investment.
\newblock {\em ACM Trans. Inf. Syst. Secur.}, 5(4):438--457, 2002.

\bibitem{Krishna}
Balachander Krishnamurthy.
\newblock I know what you will do next summer.
\newblock {\em ACM SIGCOMM Computer Communication Review}, 40(5):65--70,
  October 2010.

\bibitem{Mankiw}
N.G. Mankiw.
\newblock {\em Principles of Microeconomics}.
\newblock South-Western College Pub, 3rd edition, 2003.

\bibitem{Varian-09}
Hal Varian.
\newblock Economic aspects of personal privacy.
\newblock In William~H. Lehr and Lorenzo~Maria Pupillo, editors, {\em Internet
  Policy and Economics}, pages 101--110. Springer, 2009.

\bibitem{Newman}
M.~{Newman}.
\newblock {Power laws, Pareto distributions and Zipf's law}.
\newblock {\em Contemporary Physics}, 46:323--351, September 2005.

\bibitem{Roberts}
D.C. Roberts and D.C. Turcotte.
\newblock Fractality and self-organized criticality of wars.
\newblock {\em Fractals}, 6(4):351--357, 1998.

\bibitem{Chen}
Ye-Sho Chen, P.~Pete Chong, and Bin Zhang.
\newblock Cyber security management and e-government.
\newblock {\em Electronic Government}, 1(3):316--327, 2004.

\bibitem{Chung}
L.~Chung.
\newblock {Dealing with Security Requirements during the Developmente of
  Information Systems}.
\newblock In {\em 5th International Conference on Advanced Information Systems
  Engineering CAiSE 93}, pages 234--251, Paris France, June 1993.

\bibitem{Riederer2011}
Christopher Riederer, Vijay Erramilli, Augustin Chaintreau, Pablo Rodriguez,
  and Balachander Krishnamurthy.
\newblock {For Sale : Your Data By : You }.
\newblock In {\em Tenth ACM Workshop on Hot Topics in Networks (HotNets-X)},
  Cambridge (MA), 14-15 November 2011.

\bibitem{Gnedenko}
Boris Gnedenko and Igor Ushakov.
\newblock {\em Probabilistic Reliability Engineering}.
\newblock John Wiley \& Sons, 1995.

\bibitem{Eschenbach}
Ted~G. Eschenbach.
\newblock Spiderplots versus tornado diagrams for sensitivity analysis.
\newblock {\em Interfaces}, 22(6):40--46, Nov. - Dec. 1992.

\bibitem{Acquisti05}
Alessandro Acquisti.
\newblock Privacy (in italian).
\newblock {\em Rivista di Politica Economica}, pages 319--368, Maggio-Giugno
  2005.

\bibitem{Cramer}
J.S. Cramer.
\newblock {\em Logit Models from Economics and Other Fields}.
\newblock Cambridge University Press, 2003.

\bibitem{A29}
{Article 29 Data Protection Working Party}.
\newblock Opinion 2/2010 on online behavioural advertising.
\newblock European Commission, Directorate-General Justice, Working Paper
  WP-171, 22 June 2010.

\bibitem{FTC2010}
{Federal Trade Commission}.
\newblock {Protecting Consumer Privacy in an Era of Rapid Change}.
\newblock Technical report, FTC, December 2010.

\bibitem{Taylor2004}
Curtis~R. Taylor.
\newblock Consumer privacy and the market for customer information.
\newblock {\em The RAND Journal of Economics}, 35(4):pp. 631--650, 2004.

\bibitem{NaldiP10}
Maurizio Naldi and Andrea Pacifici.
\newblock Optimal sequence of free traffic offers in mixed fee-consumption
  pricing packages.
\newblock {\em Decision Support Systems}, 50(1):281--291, 2010.

\end{thebibliography}

\end{document}